\documentclass[11pt]{article}

\usepackage{amsmath, amsthm, amssymb, graphicx, enumerate, fullpage,amsfonts,mathrsfs,mathpazo,xspace}
\usepackage{braket}
\usepackage{color}
\usepackage{verbatim}
\usepackage{mathrsfs}
\usepackage{hyperref}

\theoremstyle{plain}
\newtheorem{thm}{Theorem}[section]
\newtheorem{theorem}[thm]{Theorem}
\newtheorem{definition}[thm]{Definition}

\newtheorem{lemma}[thm]{Lemma}

\newtheorem{claim}[thm]{Claim}

\theoremstyle{definition}

\DeclareMathOperator*{\E}{\mathbb{E}}
\newcommand{\eps}{\varepsilon}

\begin{document}

\title{A quantum lower bound for distinguishing random functions from random permutations}
\author{
Henry Yuen\thanks{Email: \texttt{\href{mailto:hyuen@mit.edu}{\color{black}hyuen@mit.edu}}. Supported by an NSF Graduate Fellowship Grant No. 1122374 and National Science Foundation Grant No. 1218547.} \\
MIT
}

\maketitle

\newcommand{\dist}{\mathcal{D}}
\newcommand{\func}{\mathcal{F}}
\newcommand{\perm}{\mathcal{P}}
\newcommand{\query}{\mathcal{Q}}
\newcommand{\polylog}{\mathrm{polylog}\, }
\newcommand{\hybrid}{\mathcal{H}}

\begin{abstract}
The problem of distinguishing between a random function and a random permutation on a domain of size $N$ is important in theoretical cryptography, where the security of many primitives depend on the problem's hardness. We study the quantum query complexity of this problem, and show that any quantum algorithm that solves this problem with bounded error must make $\Omega(N^{1/5}/\polylog N)$ queries to the input function. Our lower bound proof uses a combination of the Collision Problem lower bound and Ambainis's adversary theorem.
\end{abstract}

\section{Introduction}

We study the quantum query complexity of distinguishing between a random function and a random permutation on $N$ elements, or the \emph{RP-RF problem} (where RP stands for ``random permutation'' and RF stands for ``random function''). More precisely:
\medskip

\noindent \textbf{RP-RF problem}. Let $\dist_\perm$ denote the uniform distribution over bijective functions mapping $[N]$ to $[N]$ -- i.e. the uniform distribution over all permutations. Let $\dist_\func$ denote the uniform distribution over functions mapping $[N]$ to $[N]$.  We say an algorithm $A$ (classical or quantum) \textbf{solves the RP-RF problem on average with bias $\epsilon$}, iff
$$
	\left | \Pr_{f\sim \dist_\func} (A(f) = 1) - \Pr_{f \sim \dist_\perm} (A(f) = 1) \right | > \epsilon,
$$
where the probability is both over the distribution $\dist_\func$ or $\dist_\perm$, and the internal randomness of the algorithm $A$. We note that solving this problem is equivalent to detecting a collision in a  random function mapping $N$ elements to $N$ elements.

We work with the \emph{query model} of quantum algorithms, where the pertinent complexity measure of an algorithm is the number of queries made to the algorithm's input. As is often the case in the study of quantum query complexity, the number of computational steps is not considered. The study of quantum query complexity has been particularly fruitful. Optimal bounds on the quantum query complexity of many problems have been established, including unordered search~\cite{bennett1997strengths,grover1997quantum}, determining whether a function is $2$-to-$1$ or $1$-to-$1$ (known as the Collision Problem)~\cite{brassard1997quantum,ambainis2005polynomial}, and more. Perhaps most interestingly, there are problems for which there are exponential gaps between its classical query complexity and its quantum query complexity (e.g. Simon's Problem~\cite{simon1997power}). 

In this paper, we show that there is a quantum algorithm that solves the RP-RF problem with bias $\epsilon$ using $O(\sqrt{\epsilon} N^{1/3})$ queries, and that any such algorithm must make $\Omega(\epsilon N^{1/5}/\polylog N)$ queries. This rules out a superpolynomial speed-up from using quantum resources to solve the problem; the classical (randomized) query complexity of the RP-RF problem is $\Theta(\sqrt{\epsilon N})$~\cite{chang2008short}.

One motivation for studying the query complexity of the RP-RF problem comes from a burgeoning interest in showing that classical constructions of cryptographic primitives are secure even against quantum computers (under plausible hardness assumptions). For example, the classical Goldreich-Goldwasser-Micali construction of pseudorandom functions using pseudorandom generators~\cite{goldreich1986construct} was only recently shown to be secure against polynomial-time quantum computers by Zhandry~\cite{zhandry2012construct}. A crucial element of his result was to show that the quantum query complexity of distinguishing between functions drawn from so-called \emph{small range distributions}, and truly random functions, is large.

Zhandry posed the open question of whether one could similarly show that the classical constructions of \emph{pseudorandom permutations} (e.g., the Luby-Rackoff construction~\cite{luby1988construct}) are secure against time-bounded quantum computers~\cite{zhandry2012construct}. Pseudorandom permutations are functions on $\{0,1\}^n$ that cannot be distinguished from a truly random permutation by any algorithm running in time polynomial in $n$. They are central to the construction of block ciphers (e.g. the DES cipher): while the ideal cipher would use truly random permutations, since a truly random permutation requires an exponentially-sized description, pseudorandom permutations -- which can be specified using a polynomial number of bits -- are used instead.

An important step in the analysis of pseudorandom permutations is arguing that a classical algorithm making polynomially many queries cannot distinguish between a random function and a random permutation. If we were to try to carry the analysis over to the quantum world, this step would require that we prove a quantum query lower bound for this problem. 

Here, we show that allowing quantum queries actually does not confer a significant advantage in the number of queries needed for the RP-RF problem: we only get a modest polynomial speed-up over classical queries. This gives hope that, in a similar way to Zhandry's result, the classical constructions of pseudorandom permutations~\cite{luby1988construct,naor1999construction} can be shown to be secure against time-bounded quantum computers (assuming one way functions are secure against quantum computers).
\medskip

\noindent \textbf{Proof strategy}. The RP-RF problem is strongly reminiscent of the Collision Problem, where one has to distinguish between an $r$-to-$1$ function and a $1$-to-$1$ function, promised that one is the case. A series of works~\cite{brassard1997quantum,aaronson2004quantum,ambainis2005polynomial} established that the query complexity of the Collision Problem is $\Theta(\sqrt[3]{N/r})$.  A random function $f:[N] \to [N]$ has some predictable structure: it can be thought of as the result of an experiment where one throws $N$ labelled balls into $N$ labelled bins. The statistics of the ``balls into bins'' experiment are well studied. For example, with high probability, $f$ is $2$-to-$1$ on approximately $1/2e$ fraction of its domain, $3$-to-$1$ on approximately $1/6e$ fraction, and so on. Furthermore, with high probability, for all $y\in [N]$, $|f^{-1}(y)| \leq O(\log N/\log \log N)$. Thus, a natural inclination is to somehow embed the Collision Problem into the RP-RF problem, in order to use the Collision Problem lower bound.

We do exactly this, with some additional ideas. Using a hybrid argument, we show that any algorithm $A$ that solves the RP-RF problem with few queries must be able to distinguish between a function with an $r$-to-$1$ part, and a function without an $r$-to-$1$ part with few queries, which would contradict the Collision lower bound provided that the $r$-to-$1$ part of the function was sufficiently large. However, the Collision lower bound loses its effectiveness when this $r$-to-$1$ part is too small, so in this case we switch to using Ambainis's adversary theorem~\cite{ambainis2000quantum}. Trading off between these two lower bound techniques gives us the $\Omega(N^{1/5}/\polylog N)$ lower bound on the RP-RF problem.

Finally, it is worth noting that our lower bound also applies to the problem of \emph{detecting} a collision in a random function mapping $[N]$ to $[N]$ (i.e. finding inputs $x$, $y$ such that $f(x) = f(y)$). In fact, the query complexity of collision detection in a random function is equivalent to the RP-RF problem.

\medskip

\noindent \textbf{Subsequent work}. Shortly after this work, Zhandry was able to improve our $\Omega(N^{1/5}/\polylog N)$ lower bound to the optimal $\Theta(N^{1/3})$ query lower bound for the RP-RF problem~\cite{zhandry2013note}. In fact, he shows the more general result that detecting collisions in a random function $f:[N] \to [M]$ where $M = \Omega(N^{1/2})$ requires $\Theta(N^{1/3})$ queries. He uses techniques from~\cite{zhandry2012secure,zhandry2012construct} which are based on the idea of expressing the acceptance probabilities of quantum algorithms as low degree polynomials (also known as the polynomial method in quantum computing~\cite{beals2001quantum}). In particular, he avoids the use of the adversary method, which is known to be useless for proving lower bounds on collision-type problems~\cite{hoyer2007negative} (this is also the reason why it is unlikely that our proof can be extended to obtain a lower bound better than $\Omega(N^{1/5})$).



%

\section{Preliminaries}

Before analyzing the query complexity of the RP-RF problem, we introduce some terminology that will be useful later.

\begin{definition}
	Let $f:[N] \to [N]$. We say that a domain element $x\in [N]$ \textbf{has multiplicity $u$ under $f$} iff $|f^{-1}(f(x))| = u$.
\end{definition}


\begin{definition}
	Let $f:[N] \to [N]$. The \textbf{collision profile} of $f$, denoted $\mathrm{Col}(f)$, is a vector of non-negative integers $\langle b_1,\ldots,b_N\rangle$ such that there exists a partition of $[N] = S_1 \cup \cdots \cup S_N$, where for all $i$, 
	\begin{enumerate}
		\item $|S_i| = b_i$ (some $S_i$ may be empty), and
		\item For all $x\in S_i$, $x$ has multiplicity $i$ under $f$.
	\end{enumerate}
\end{definition}

For example, the collision profile of a permutation $f$ is $\mathrm{Col}(f) = \langle N, 0,\ldots,0 \rangle$. More generally, the collision profile of a $r$-to-$1$ function is the vector $N \cdot e_r$, where $e_r$ is the $r$th elementary basis vector. When we write $\dist_C$ for some collision profile $C$, we mean the uniform distribution over functions with collision profile $C$.

\newcommand{\maxload}{\mathrm{maxload}}

\begin{definition}
	Let $C = \langle b_1,\ldots,b_N \rangle$ be a collision profile for a function $f:[N] \to [N]$. The \textbf{maximum load} of $C$, denoted $\maxload(C)$, is the largest $i$ such that $b_i \neq 0$. 
\end{definition}

Let $f,g:[N] \to [N]$. We say that $f$ and $g$ are \textbf{collision equivalent} iff $\mathrm{Col}(f) = \mathrm{Col}(g)$. It is easy to see that the relation induced by collision equivalence is an equivalence relation on the set of all functions mapping $[N]$ to $[N]$. We will generally identify the equivalence classes with their corresponding collision profiles. We will abuse notation and write $f \in C$ to denote a function with collision profile $C$. For a quantum algorithm $A$, we will use $\query(A)$ to denote its quantum query complexity.

\begin{definition}[Good profile]
Let $C = \{b_1,\ldots,b_N\}$ be a collision profile for functions mapping $[N]$ to $[N]$. $C$ is \textbf{good} iff $\maxload(C) < 3 \log N/\log \log N$.
\end{definition}

\begin{claim}
\label{clm:good_profile_likely}
Let $f$ be a random function drawn from $\dist_\func$. Then $\Pr(\text{$\mathrm{Col}(f)$ is good}) \geq 1 - 1/N$.
\end{claim}
\begin{proof}
	This follows from standard ``balls in bins'' arguments (see, e.g., Lemma 5.1 in~\cite{mitzenmacher2005probability}, for a proof).
\end{proof}

\section{An upper bound}

We show that allowing quantum queries gives a \emph{slight} advantage of classical queries for the RP-RF problem. The classical query complexity of solving the RP-RF problem with bias $\epsilon$ is $\Omega(\sqrt{\epsilon N})$ -- see~\cite{chang2008short} for a proof.

\begin{theorem}
	The quantum query complexity of solving the RP-RF problem on average with bias $\epsilon$ is $O(\sqrt{\epsilon} N^{1/3})$.
\end{theorem}
\begin{proof}
	We use the the collision-finding algorithm of Brassard, Hoyer, Tapp~\cite{brassard1997quantum}. The algorithm works as follows:
	\begin{itemize}
		\item Pick a set $T \subseteq [N]$ of size $k = N^{1/3}$ at random. Query the value of $f(x)$ for every $x \in T$, and store the values $(x,f(x))$ in a table $L$.
		\item If there exists distinct $x,y\in T$ such that $f(x) = f(y)$, output ``1''.
		\item Construct the oracle $H(x)$ such that $H(x) = 1$ iff there exists a $s \in T$ such that $f(x) = L(s)$, where $L(s)$ denotes the stored value $f(s)$ in $L$.
		\item Execute Grover's algorithm on oracle $H$ to find an $x\in [N]$ such that $H(x) = 1$. If there exists one, output ``1''. Otherwise, output ``0'' (for random permutation). 
	\end{itemize}
		
	Clearly, if the input function is a random permutation, then the algorithm will always output ``0''. Suppose that $f$ were drawn from $\dist_\func$. With high probability, there exists a set $S \subseteq [N]$ of size $\approx (1 - 1/e)N$ such that for every $x\in S$, there exists a $y\neq x$ in $S$ where $f(x) = f(y)$. 
	
	If there were a collision in $T$, we are done; computing the table $L$ requires $O(N^{1/3})$ queries to $f$. Otherwise, the number of $x\in [N]$ such that $H(x) = 1$ is, with high probability, approximately $|S| \cdot |T| /N = \Omega(|T|) = \Omega(N^{1/3})$. Thus, Grover's algorithm will find a collision with probability $\Omega(\epsilon)$ in $O(\sqrt{\epsilon} N^{1/3})$ queries~\cite{boyer1996tight}, and the collision-finding algorithm will output ``1''. 
\end{proof}

\section{An average-case to worst-case reduction}

We first show that any efficient algorithm $A$ that can solve the RP-RF problem on average can be transformed into an efficient algorithm $B$ that solves a restricted form of the RP-RF problem, but in the \emph{worst case}. This ``restriction'' is that $B$ is only guaranteed to distinguish between a random permutation, and a random function with a \emph{specific} collision profile. 

\begin{claim}[Average case to worst case]
\label{clm:avg_to_worst_case}
Let $A$ be a quantum algorithm that solves the RP-RF problem on average, with bias $\epsilon > 2/N$. Then there exists an a good collision profile $C$, and a quantum algorithm $B$ that makes $O(\frac{1}{\epsilon} \query(A))$ queries, such that
\begin{enumerate}
	\item For all $f\in C$, $\Pr(B(f) = 1) \geq 2/3$, and
	\item For all $f\in \perm$, $\Pr(B(f) = 1) \leq 1/3$.
\end{enumerate}
In this case, we say that the algorithm $B$ distinguishes between $C$ and $\perm$, with bounded error.
\end{claim}
\begin{proof}
	Without loss of generality, assume that $\Pr_{f\sim \dist_\func} (A(f) = 1) - \Pr_{f \sim \dist_\perm} (A(f) = 1) > \epsilon$ (i.e. without the absolute value bars). We can rewrite this as
	$$
		\E_{R} \E_{f \sim \dist_R} [A(f)] - \E_{f \sim \dist_\perm} [A(f)] > \epsilon,
	$$
	where $R$ is a random collision profile, distributed so that sampling $R$, then sampling a function $f\sim \dist_R$ is equivalent to sampling $f$ from $\dist_\func$ directly. Let $p := \E_{R} \E_{f \sim \dist_R} [A(f)]$. Define $\mathcal{R}$ to be the set of collision profiles such that for all $R \in \mathcal{R}$, $\E_{f \sim \dist_R}[A(f)] \geq p - \epsilon/2$. Then, we have that
	$$
		\Pr(\mathcal{R}) \cdot 1 + (1 - \Pr(\mathcal{R})) \cdot (p - \epsilon/2) \geq p.
$$
	Since $p \geq \epsilon$, it must be that $\Pr(\mathcal{R}) \geq \epsilon/2$. By Claim~\ref{clm:good_profile_likely}, there must exist a profile $C\in\mathcal{R}$ that is good, because $\Pr_C(\text{$C \in \mathcal{R}$, $C$ is good}) \geq \Pr(\mathcal{R}) - \Pr_C(\text{$C$ is not good}) \geq \epsilon/2 - 1/N > 0$. Fix such a good $C$. Thus,
	$$
		\E_{f \sim \dist_C}[A(f)] - \E_{f\sim \dist_\perm}[A(f)] > \epsilon/2,
	$$
	where $\dist_C$ denotes the uniform distribution over functions with collision profile $C$. Now consider the following algorithm $B'$, which takes a function $f:[N] \to [N]$ as input:
	\begin{itemize}
		\item Pick random and independent permutations $\pi,\sigma$ on $[N]$.
		\item Simulate $A$, and whenever $A$ makes a query to the $x$th element of its input, for some $x\in [N]$, return $\pi(f(\sigma(x)))$.
		\item Return the output of $A$.
	\end{itemize}
	
	Suppose that the input to $B'$ were a function $f \in C$. Then, for random permutations $\pi,\sigma$, $\pi(f(\sigma))$ is uniformly distributed over $C$. This is because set of permutation pairs $(\pi,\sigma) \in S_N \times S_N$ that leave a function $f:[N] \to [N]$ invariant forms a subgroup. If $f \in \perm$, then $f$ is a uniformly random permutation. Thus, for all $f \in C$, $\Pr(B'(f) = 1) = \E_{f \in \dist_C}[A(f)]$, and for all $f \in \perm$, $\Pr(B'(f) = 1) = \E_{f \in \dist_\perm}[A(f)]$. Thus, $B'$ is an algorithm that makes $\query(A)$ queries and distinguishes between an arbitrary function $f$ in $C$ and an arbitrary permutation in $\perm$ with bias $\epsilon/2$. Using standard amplitude estimation techniques~\cite{boyer1996tight}, from $B'$ we obtain an algorithm $B$ that distinguishes between $C$ and $\perm$ in the worst case, with bounded error, making $O(\frac{1}{\epsilon} \query(A))$ queries.
\end{proof}

\section{The lower bound proof}

We now prove that any algorithm $B$ that distinguishes between $C$ and $\perm$ for some good collision profile $C$ must make a large number of queries. We do this by exhibiting $O(\log N)$ hybrid distributions $\hybrid_0,\ldots,\hybrid_r$ such that $\hybrid_0 = \dist_\perm$, $\hybrid_r = \dist_C$, and argue that distinguishing between any two adjacent hybrids requires large query complexity.

Let $C = \langle c_1,\ldots,c_N \rangle$. Partition $[N]$ as $\{1\} \cup I_{\mathrm{empty}} \cup I_{\mathrm{small}} \cup I_{\mathrm{large}}$, where 
\begin{enumerate}
	\item $I_{\mathrm{empty}} = \{i : c_i = 0,\, i > 1 \}$,
	\item $I_{\mathrm{small}} = \{i : c_i < N^d,\, i > 1 \}$, and
	\item $I_{\mathrm{large}} = \{i : c_i \geq N^d,\, i > 1\}$,
\end{enumerate}
for some parameter $d\in (0,1)$ that we will optimize later. Write $I_{\mathrm{large}} = \{i_1,\ldots,i_q\}$.  Let $r = q + 1$, and define the following sequence of collision profiles:
\begin{itemize}
	\item $H_0 = \langle N, 0, \ldots, 0\rangle$ (i.e., the collision profile of a permutation).
	\item For all $1 \leq j \leq q$, $H_j$ is the collision profile $\langle h_{j1},\ldots,h_{jN}\rangle$, where $h_{j1} = N - (c_{i_1} + \cdots + c_{i_j})$, $h_{ji_k} = c_{i_k}$ for all $k = 1,\ldots,j$, and otherwise $h_{ji} = 0$,
	\item $H_{q+1} = C$.
\end{itemize}

We then define the hybrid distribution $\hybrid_j$ to be the uniform distribution over $H_j$ for each $j$. Note that, since $C$ is good, $\maxload(C) = \max I_{\mathrm{small}} \cup I_{\mathrm{large}} < 3 \log N/\log \log N$. Thus, the number of hybrids is $q+2 = O(\log N)$.

In other words, starting with the permutation collision profile, we gradually ``morph'' the hybrid collision profiles into the final collision profile $C$. We do so by adding the $c_i$-to-$1$ parts (for large $c_i$), one by one, to the hybrid profiles, until we only have small $c_i$ left. We then add all the small $c_i$ together to the last hybrid profile at once.

\begin{claim}
\label{clm:hybrid}
	Let $B$ be a quantum algorithm that distinguishes between $C$ and $\perm$ with bounded error. Then there exists an $1 \leq j \leq q+1$, and a quantum algorithm $B'$ that makes $O\left(q \query(B) \right)$ queries, and distinguishes between $H_{j-1}$ and $H_j$ with bounded error.
\end{claim}
\begin{proof}
	Since $B$ distinguishes between $C$ and $\perm$, $B$ also distinguishes between $\dist_C$ and $\dist_\perm$ in an average-case sense. Then, by the triangle inequality, we have that there exists an $1 \leq j \leq q+1$ such that
	$$
		\left  | \E_{f \sim \hybrid_{j-1}} [B(f)] - \E_{f\sim \hybrid_j} [B(f)] \right | \geq \Omega\left(\frac{1}{q} \right).
	$$
	In a similar way to the proof of Claim~\ref{clm:avg_to_worst_case}, we then obtain the desired algorithm $B'$.
\end{proof}

We divide the analysis into two cases: $j < q+1$ and $j = q+1$. 

\begin{lemma}
\label{lem:large}
	If $j < q+1$, then $\query(B') = \Omega \left( \left(\frac{N^{d}}{\log N} \right)^{1/3} \right)$.
\end{lemma}
\begin{proof}
	Since $j < q+1$, the hybrid collision profiles $H_{j-1}$ and $H_{j}$ differ by one ``collision type'' $c_{i_j}$. That is, $H_j$ contains a $i_j$-to-$1$ part, and $H_{j-1}$ doesn't. We now reduce the $i_j$-to-$1$ collision problem to the problem of distinguishing between $H_{j-1}$ and $H_j$. We will ``embed'' a function that is either $c_{i_j}$-to-$1$ or $1$-to-$1$ into a function from either $H_j$ or $H_{j-1}$, respectively. 
	
	Define $g_k = c_{i_1} + \cdots + c_{i_k}$, for $k = 1,\ldots,j$. We will define $g_0 = 0$. Let $f$ be a function mapping $[c_{i_j}]$ to $[c_{i_j}]$. Now let $h_f:[N] \to [N]$ be defined as follows: 
	\begin{enumerate}
		\item For all $1 \leq k < j$, $h_f$ restricted to $\{g_{k-1}+1,\ldots,g_k\}$ is an arbitrarily chosen $i_k$-to-$1$ function with range that's a subset of $\{g_{k-1}+1,\ldots,g_k\}$.
		\item For all $x\in \{g_{j-1}+1,\ldots,g_j\}$, $h_f(x) = f(x) + g_{j-1}$.
		\item For all $x\in \{g_j + 1,\ldots,N\}$, $h_f(x) = x$.
	\end{enumerate}
	
	Consider the following algorithm $G$, on input $f:[c_{i_j}]\to [c_{i_j}]$: 
	\begin{itemize}
		\item Simulate algorithm $B'$ on input $h_f$.
		\item Return the output of $B'$.
	\end{itemize}
	
	Observe that $\query(G) = \query(B')$. Now, if $f$ were a $1$-to-$1$ function, then $h_f \in H_{j-1}$. Otherwise, if $f$ were a $i_j$-to-$1$ function, then $h_f \in H_j$. Since $B'$ distinguishes between $H_{j-1}$ and $H_j$, $G$ distinguishes between $i_j$-to-$1$ functions and $1$-to-$1$ functions. By the collision lower bound of Aaronson, Shi, and Ambainis~\cite{aaronson2004quantum,ambainis2005polynomial}, $G$ (and hence $B'$) must make $\Omega \left( (c_{i_j}/i_j )^{1/3} \right)$ queries. Since $j < q+1$, we have that $c_{i_j} \geq N^d$ and since $C$ is good we have $i_j < \log N$.	
\end{proof}

\begin{lemma}
\label{lem:small}
	If $j = q+1$, then $\query(B') = \Omega \left( \sqrt{\frac{N}{N^d \log N}} \right)$.
\end{lemma}
\begin{proof}
We adapt the proof of Lemma 16 from~\cite{aaronson2009need} (which is an application of Ambainis's adversary theorem) to this situation, to show that distinguishing between $H_q$ and $H_{q+1}$ requires at least $\Omega \left( \sqrt{\frac{N}{N^d \log N}} \right)$ queries.

Let $I_\mathrm{small} = \{w_1,\ldots,w_r\}$. These are all the integers $u$ such that the number of domain elements in $[N]$ that have multiplicity $u$ is less than $N^d$. By definition of $I_{\mathrm{small}}$, we obtain the bound $v := c_{w_1} + \cdots + c_{w_r} < \maxload(C) N^d < N^d \log N$. Define a relation $R \subseteq H_{q+1} \times H_{q}$, where $(f_1,f_2) \in R$ iff one can transform $f_1$ to $f_2$ in the following way:
\begin{enumerate}
	\item Let $S = \{x\in [N] :\text{ $x$ has multiplicity $w_s$ under $f_1$ for some $1 \leq s \leq r$} \}$.
	\item Define an intermediate function $\tilde{f}$ that is identical to $f_1$ except on $S$, $\tilde{f}(S) \subseteq f_1(S) \cup ([N] - f_1([N]))$, and $|\tilde{f}(S)| = |S|$ (i.e. $\tilde{f}$ is $1$-to-$1$ on $S$, or that $\tilde{f} \in H_q$).
	\item $f_2$ is equal to $\tilde{f}$, except for every $x\in S$, there exists a unique $y\in [N] - S$ such that $f_2(y) = \tilde{f}(x)$ and $f_2(x) = \tilde{f}(y)$.
\end{enumerate}
Equivalently, $(f_1,f_2)\in R$ iff there is a transformation from $f_2$ to $f_1$ by a similar procedure:
\begin{enumerate}
	\item Let $S \subseteq \{x \in [N]:\text{ $x$ has multiplicity $1$ under $f_2$}\}$ be such that $|S| = c_{w_1} + \cdots + c_{w_r} = v$.
	\item Define an intermediate function $\tilde{f}$ that is identical to $f_2$ except on $S$, $\tilde{f}(S) \subseteq f_2(S)$, and $\tilde{f} \in H_{q+1}$.
	\item $f_1$ is equal to $\tilde{f}$, except for every $x\in S$, there exists a unique $y \in [N] - S$ such that $f_1(y) = \tilde{f}(x)$ and $f_1(x) = \tilde{f}(y)$.
\end{enumerate}
By an analysis similar to that in~\cite{aaronson2009need}, using Ambainis's adversary theorem, we get that $\query(B') = \Omega \left( \sqrt{\frac{N}{v}} \right) = \Omega \left( \sqrt{\frac{N}{N^d \log N}} \right)$.
\end{proof}

\begin{theorem}
	Let $A$ be a quantum algorithm that solves the RP-RF problem on average, with bias $\eps > 2/N$. Then $\query(A) = \Omega \left(\frac{\epsilon  N^{1/5}}{\log^{3/2} N} \right)$.
\end{theorem}
\begin{proof}
	We balance the lower bounds from Lemmas~\ref{lem:large} and~\ref{lem:small} by setting $d = 3/5$, which gives us that $\query(B') = \Omega \left( N^{1/5} / \sqrt{\log N} \right)$. Then, since Claim~\ref{clm:hybrid} gives us that algorithm $\query(B') = O(q\query(B))$, $B$ must make $\Omega \left(\frac{N^{1/5}}{\log^{3/2} N} \right)$ queries (because $q = O(\log N)$), and Claim~\ref{clm:avg_to_worst_case} gives that algorithm $A$ must make $\Omega \left(\frac{\epsilon N^{1/5}}{\log^{3/2} N} \right)$ queries.
\end{proof}

%
%

\textbf{Acknowledgments}. We thank Scott Aaronson for a number of insights and ideas for this problem. We also thank the anonymous referees for their useful comments, and for noticing an error in an earlier proof of Claim~\ref{clm:avg_to_worst_case}. 

\bibliographystyle{plain}
\bibliography{rf_rp_lower_bound}

\end{document}